\documentclass{article}

\usepackage{xcolor}
\usepackage{arxiv}
\usepackage[title]{appendix}
\usepackage[utf8]{inputenc} % allow utf-8 input
\usepackage[T1]{fontenc}    % use 8-bit T1 fonts
\usepackage{hyperref}       % hyperlinks
\usepackage{url}            % simple URL typesetting
\usepackage{booktabs}       % professional-quality tables
\usepackage{amsfonts}       % blackboard math symbols
\usepackage{nicefrac}       % compact symbols for 1/2, etc.
\usepackage{microtype}      % microtypography
\usepackage{lipsum}
\usepackage{graphicx}
\graphicspath{ {./images/} }

\usepackage{amsthm}

\newtheorem{theorem}{Theorem}
\newtheorem{lemma}[theorem]{Lemma}
\newtheorem{problem}[theorem]{Problem}
\usepackage{amsmath}
\usepackage{enumerate}
\usepackage{algorithm}
\usepackage[noend]{algpseudocode}
\usepackage{ amssymb }
\usepackage{comment}

\usepackage{amsmath}
\DeclareMathOperator*{\argmax}{argmax}

\newcommand{\R}{\mathbb{R}}
\newcommand{\at}[2][]{#1|_{#2}}

\title{Incentivizing Federated Learning}

\author{
 Shuyu Kong \\
  Department of Electrical and Computer Engineering\\
  Northwestern University\\
  Evanston, IL 60201 \\
  \texttt{shuyukong2020@u.northwestern.edu} \\
   \And
 You Li \\
  Department of Electrical and Computer Engineering\\
  Northwestern University\\
  Evanston, IL 60201 \\
  \texttt{you.li@u.northwestern.edu} \\
  \And
 Hai Zhou \\
  Department of Electrical and Computer Engineering\\
  Northwestern University\\
  Evanston, IL 60201 \\
  \texttt{haizhou@northwestern.edu} \\
}

\begin{document}

\begin{comment}
    simplify: utility function,
    price of anarchy.
    Must incentivize, otherwise, will not come.
    How to incentivize.
        two cases: 1. no validation dataset
                   2. has validation dataset.
                   
    After each case, build simplified model and prove something. (even though simplified, convey essential idea why my mechanism should work)
    
    data market segregation. "big player" want to partner with "big player".
    
    Under what condition, models are segregated.
    
    I'm going to use a extremely simplified model to convince.
                   
\end{comment}

\maketitle
\begin{abstract}
Federated Learning is an emerging distributed collaborative learning paradigm used by many of applications nowadays. The effectiveness of federated learning relies on clients' collective efforts and their willingness to contribute local data. However, due to privacy concerns and the costs of data collection and model training, clients may not always contribute all the data they possess, which would negatively affect the performance of the global model.

This paper presents an incentive mechanism that encourages clients to contribute as much data as they can obtain. Unlike previous incentive mechanisms, our approach does not monetize data. 
Instead, we implicitly use model performance as a reward, i.e., significant contributors are paid off with better models. We theoretically prove that clients will use as much data as they can possibly possess to participate in federated learning under certain conditions with our incentive mechanism.

\end{abstract}
\section{Introduction}

Deep learning is data-hungry. Nevertheless, a single party may not have sufficient data for many deep learning training tasks. For example, a hospital that wants to train a disease diagnosis classifier may not have enough cases. Federated learning (FL) is brought up to leverage the strength of massive data distributed among different clients securely and privately. The clients who participate in federated learning collaboratively complete a deep learning task under the coordination of a centralized server or a service provider called the federator. The raw data of each client is stored locally and not transferred. Instead, model parameters or model updates are exchanged to complete the task.

It appears that FL avoids data sharing among clients. However, in the basic FL setting, the centralized server will distribute the same global model to all the clients. As a result, those clients who make only minimal contributions benefit in the same way as those who make significant contributions. Therefore, clients in practice may intentionally contribute fewer or even no data in FL to reduce local training costs and privacy risks. The performance of the global model is only slightly affected if only a few clients "cheat" in this way. 
Nevertheless, suppose all clients have such "selfish" mindsets that drive them to contribute little while hoping to gain a lot. In that case, there is no way the centralized server could produce a well-performed global model on low-quality training data.
Consequently, every client ends up obtaining a worse model, making the whole FL process meaningless.

Recently, many incentive mechanisms have been proposed to incentivize clients to contribute more data~\cite{sarikaya2019motivating, Zhan20Learning, feng2018joint, Zeng_2020, jiao2020automated,Le20AucCellular}. However, almost all those incentive mechanisms use money as a reward. 
We believe payments could be obstacles for many applications. Furthermore, the value of high-performance models is not always measurable, and many clients participate in seek of a model instead of money. For those clients, it is hard to compensate for losses on model performance with money.

This paper proposes a non-monetary incentive mechanism that uses model performance as a reward. In our mechanism, the server will evaluate the performance of every client's uploaded model in each round and distribute different models to clients based on the evaluation results. With the assumption that clients who contribute better quality data will upload a better local model, the objective of our incentive mechanism is to reward significant contributors with better aggregated models. We prove that our approach satisfies the baseline incentive requirements and can drive clients to contribute the maximum amount of data under certain conditions.

\section{Related Work}

\subsection{Incentive Mechanisms in Federated Learning}

Recently, various techniques including Stackelberg game, auction, contract, Shapley value, blockchain and reinforcement learning, have been adopted as the incentive
mechanisms for federated learning.
Stackelberg game~\cite{sarikaya2019motivating, Zhan20Learning, feng2018joint}, auction~\cite{Zeng_2020, jiao2020automated,Le20AucCellular},
and contract~\cite{kang2019incentive, Ding20MultiDimensional} are mainly employed for node
selection and payment allocation, while Shapley value~\cite{Wnag19MeasureContribution, Song19ProfitAllocation} is used for contribution measurement. Both
blockchain~\cite{liu2020fedcoin} and reinforcement learning~\cite{jiao2020automated, Wang20OptimizingNonIID} are auxiliary techniques to improve performance and
robustness.

\subsection{Fairness in Federated Learning}

\begin{comment}
There are different aspects of fairness.

Federated learning is an attractive paradigm for fitting a model to data generated by, and residing on, a network of remote devices.
\end{comment}

Fairness for federated learning is another close topic to our research. Li~\cite{li2020fair} proposes q-FFL, a novel optimization objective that addresses fairness issues in federated learning. Inspired by fair resource allocation for wireless networks, q-FFL minimizes an aggregate reweighted loss parameterized by q, such that the devices with higher loss are given higher relative weight. Authors claim that q-FFL maintains the same overall average accuracy while
ensuring a more fair and uniform service quality across
the network. Agnostic Federated Learning~\cite{mohri2019agnostic} is proposed to resolve the trade-off between performance and fairness~\cite{huang2020fairness}.

Those approaches allocate appropriate weights in aggregating client models to simultaneously achieve fairness and performance. Hence, they still lie in the vanilla federated learning category and do not necessarily encourage clients to provide more data. Compared with these studies, our work takes one step further in the sense that we achieve fairness in learning and leverage fairness to incentivize all the clients to contribute as much data as they possess. 
    
\subsection{Privacy Threat in Federated Learning}
    
\begin{comment}
There is always security cost involved in federated learning. Several works ~\cite{NEURIPS2019_60a6c400, zhao2020idlg}model can be stolen from gradient leakage.
deep leakage ~\cite{NEURIPS2019_60a6c400}
imporved deep leakage ~\cite{zhao2020idlg}
\end{comment}

Several earlier studies demonstrate that attackers could leverage the model information uploaded by clients, e.g., gradient, to infer private data~\cite{NEURIPS2019_60a6c400, zhao2020idlg}. This imposes a privacy threat on those clients whose data is confidential. As a result, clients may not fully commit to contributing all their data. This phenomenon has motivated us to develop our incentive mechanism.
%This phenomenon also motivates us to develop an incentive mechanism to reward significant contributors to compensate for the potential loss so that clients will contribute more data compared with under the traditional FL setting.

\subsection{Personalization of Federated Learning}

While vanilla federated learning aims to train a single global model across decentralized local datasets, a monolithic model may not fit all clients. Federated personalization allows a client to obtain a more potent model per the client's specific objective. Rather than averaging with constant weights for the entire federation, optimal weighted model combinations are computed for each client after figuring out how much it can benefit from every other client's model. 

Personalization in FL shares a similar methodology with our incentive mechanism - different clients are distributed with different models - but with distinct objectives. Typical personalization aims to adapt models to the client's data distribution to obtain a well-performing model on its data of interest. On the other hand, our incentive mechanism attempts to use performance discrepancy to encourage clients to contribute more data. Recent personalization work~\cite{huang2020personalized, zhang2020personalized} introduces an attention-inducing function to facilitate collaborations between similar clients. In contrast, our approach uses weighted aggregation to ensure that models with similar performance can access similar data,
%closely collaborate to boost their performances, 
thus achieving our objective.
%thus achieving the objective that the client uploading a better model should obtain a better aggregated model from server as well.

\begin{comment}

    hierarchical federated learning?: give different clients different models. approach similar with us, but with different objectives.
    
    We do not assume knowledge of any underlying data distributions or client similarities, and allow each client to optimize for arbitrary target distributions of interest, enabling greater flexibility for personalization.
    
    \textcolor{red}{~\cite{huang2020personalized}~\cite{zhang2020personalized}}

\end{comment}

\section{Background and Problem Formulation}

\subsection{Background}
%We first introduce the common notations used for federated learning. 
Federated learning is a clever way of utilizing disjointed data and computational resources to train machine learning models.
We refer to the set of all clients as $\mathcal{C}$. $\mathcal{D}_k$ denotes the local dataset of client $\mathcal{C}_k \in \mathcal{C}$, which consists of $\lvert \mathcal{D}_k \rvert$ samples. $\mathcal{D} = \bigcup_{k\in 1\cdots \lvert \mathcal{C} \rvert} \mathcal{D}_k$ represents the full training set. 
%For an arbitrary client $\mathcal{C}'$,
$\mathcal{L}(w, \mathcal{D}_k) = \frac{1}{\lvert \mathcal{D}_k \rvert} \sum_{z\in \mathcal{D}_k} \mathcal{L}(w, z)$ denotes the empirical loss over model $w$ and $\mathcal{D}_k$. In the $t$-th round of federated learning, the optimization problem is formulated as minimizing the empirical loss over all the training examples of the clients $\mathcal{C}_t$ who participate in that round:

\begin{equation}
    \min\limits_{w\in \R_p} \mathcal{L}(w, \mathcal{D}) = \sum\limits_{k\in \mathcal{C}_t} \frac{\lvert \mathcal{D}_k \rvert}{N(\mathcal{C}_t)} 
    \mathcal{L}(w, \mathcal{D}_k)
\end{equation}

where $N(\mathcal{C}_t)$ represents the total number of samples belonging to the clients in $\mathcal{C}_t$. In other words, the global loss function $\mathcal{L}(w, \mathcal{D})$ is the weighted average of local functions $\mathcal{L}(w, \mathcal{D}_k)$ and the weights are proportional to the sizes of the local datasets. We adopt a widely used standard FL algorithm, federated averaging (FedAvg)~\cite{mcmahan2017communication}, to solve the optimization problem. FedAvg is the basic building block of our algorithm and it executes as the following. In the initial stage, the centralized server randomly initializes a global model $w_0$. Then the training process orchestrates alternated local and global updates and communications between the server and the clients in every round $t \in 1 \cdots T$. Specifically, each client performs local training for a few epochs during each round. Afterward, the local models are uploaded online so that the server can perform federated aggregation to generate a single global model. Subsequently, the server distributes the global model to the clients and proceeds to the next round. The training, uploading and distributing processes are repeated until the global model converges. 

Formally speaking, during the $i$-th training epoch of the $t$-th round, each client updates its local model through gradient descent:

\begin{equation}
    w^k_{t,i}= w^k_{t,i-1} -  \eta \nabla \mathcal{L}(w^k_{t,i-1}, \mathcal{D}_k ),
\end{equation}

and at the end of the $t$-th round, the server performs federated aggregation to compute the weighted or unweighted average of selected client models:

\begin{equation}
    w_t \leftarrow \sum\limits_{k\in \mathcal{C}_t} \frac{\lvert \mathcal{D}_k \rvert}{N(\mathcal{C}_t)} w^k_{t,m}.
\end{equation}

The above weighted aggregation assumes that the information of the local data is known, which might be unrealistic in practice. On the contrary, unweighted aggregation is performed when such information is unknown:

\begin{equation}
    w_t \leftarrow \sum\limits_{k\in \mathcal{C}_t} w^k_{t,m}.
\end{equation}

\begin{comment}
    In the traditional federated learning, all the clients are distributed with the same model even though there might be large performance discrepancy between client local models. Therefore, some clients who do not actually want to contribute their data could potentially upload a garbage model with almost none data used for training while still downloading a good model from the server. By making contribution in the federate learning, the clients need to collect data and perform local training epochs over the collected data, both of which are not free and require costly human efforts and computation resources. Therefore, a client who expect others to make maximum data contribution could upload a random model to save participation cost while obtain a model with minimum performance loss. But if all the clients decide to do so, no one would get a meaningful model from the server. Our goal is to avoid such scenario.
\end{comment}

Please refer to Table~\ref{tab:symbol} for a list of notations used in this paper.

The effectiveness of vanilla federated learning relies on clients' collective efforts and their willingness to contribute local data. However, we believe the unselfish assumption is unrealistic when considering the moral hazard. In reality, every client could contribute very little data or even upload a fraudulent model when participating in federated learning. Once it receives the global model, it fine-tunes that model with all the local data. In this way, the client could get almost all the benefits while devoting the least amount of data. Federated learning will be downgraded to individual learning if most clients take this strategy.
%If every client does that, the federation becomes meaningless since no one would get a good model from centralized server. As a result, the federated learning becomes similar to individual machine learning.

\subsection{Problem Formulation}

It is paramount to develop an incentive mechanism to avoid such a scenario.
%to encourage all the clients to participate with as much data as possible.
To address the incentive problem, we first derive the utility function with respect to a client's local data contribution. In the vanilla FL framework, it is assumed that all clients will always contribute all local data. Therefore, a client's utility is only related to the performance of the final model distributed by the centralized server, i.e., $u_i = \gamma p(d_i)$. As already discussed, this is usually not the case.
In contrast, our framework assumes that every client has a goal to maximize its utility by choosing the proper amount of data they would collect or devote when participating in federated learning.
%Such modeling is the root reason to lead to the improper belief that client would make maximum data contribution. 
%such belief, when broken, could make federated learning a meaningless game. 
Hence, along with the model performance, we add another term to the client's utility function. This term evaluates the cost of a client's participation, including the efforts of data collection and the consumption of computational resources on local training.
\begin{equation}
    u_i = \gamma p(d_i, D_{\overline{i}}) - \alpha c(d_i).
\end{equation}

Here, $d_i$ represents the local data possessed and contributed by $\mathcal{C}_i$ and $D_{\overline{i}}$ represents data possessed by all the remaining clients and is available to $\mathcal{C}_i$. Accordingly, $p_i = p(d_i, D_{\overline{i}})$ denotes the performance of the model sent to client $i$ while $c_i = c(d_i)$ denotes the client's participation cost. 
%$\gamma$ and $\alpha$ are coefficients. 
We can then formulate the data contribution decision problem as a utility maximization problem:

\begin{problem}[Utility Maximization for Federated Learning]
Design an incentive mechanism and corresponding utility functions that encourages all clients to contribute as much data they can possibly obtain:

\begin{equation}
   \forall i:  \argmax_{d_i} u_i(d_i) = d^t_i.
\end{equation}
\end{problem}

In our setting, rewards are not monetized but reflected by the quality of the models the server sends back to clients. We do not require clients to report their training data sizes honestly. Otherwise, the problem becomes trivial as the server can easily compute the contribution levels, based on which the centralized server rewards the significant contributors.
%The clients who use more data to train a local model and participate in the federated learning should be rewarded with a better model after aggregation from the server.

\begin{comment}
    (if fake data, could that negatively affect the model performance? how should we model it in our framework.)
\end{comment}

To facilitate our design, we make the following reasonable \textit{Assumptions}:

\begin{enumerate}[I.]
\item \label{assump:1} We assume a maximum threshold, $d^t_i$, of the amount of data the client $i$ can obtain.
\item \label{assump:2} We assume that model quality is completely determined by data quality for a specific FL training task. For simplicity, we further assume that data quality is solely decided by the amount of data used in training. We let a fixed degeneration factor, $\gamma$, on the model quality for federated aggregation, compared with the case that a single learner can utilize all data in $\mathcal{D}$. %: $p_f(\sum d_i) = \gamma p_T(\sum d_i)$. We assume both $p(d_i)$ and $c(d_i)$ monotonically increase as $d_i$ increases. 
\item \label{assump:3} We assume the centralized server has a validation dataset so that it can approximately evaluate a local model's performance on $\mathcal{D}$.
\item \label{assump:4} $p(d_i)$ is concave because the marginal performance gain decreases with respect to the amount of data used for training when the dataset is sufficiently large. On the other hand, $c(d_i)$ is convex because the marginal cost to obtain new data increases when the dataset is already sufficiently large.
\end{enumerate}

\section{The Incentivized Federated Learning Algorithm}

Our main idea is to reward a client who makes a more significant contribution with a higher performance model, aggregated from a broader range of other clients' models. More specifically, the centralized server ranks models uploaded by all clients based on their performances. Then a client's model will be aggregated with those models with lower rankings and then sent back to the client. Formally, the aggregation process can be stated as: 
\begin{equation}
    w^k_t \leftarrow \sum\limits^{r(k)}_{j = 1} w^j_{t,m},
\end{equation}

where $r(k)$ represents the position of client $k$ after ranking. 
%That means all the models worse than the models from client $k$ will be used to generate the aggregated model for client $k$. By definition, the number of models for aggregation $r(k)$ is decided by the position in the ranked performance, that is, if the client model has relatively good performance and ranked higher among all, the $r(k)$ will be larger, and vice versa. This is to ensure the "winner" would be paid with more "reward". 
Note that there is no global model since every client will be distributed a different model based on the performance of its own model. 

Our proposed incentivized federated learning framework is illustrated by Algorithm~\ref{Alg:Iter}. Essentially, we ensure that $p(d_i, D_{\overline{i}})$ is depending on both $d_i$ and $D_{\overline{i}}$, while $d_i$ can also impact $D_{\overline{i}}$.
To realize this effect, we assume the centralized server has access to a small validation dataset that follows the global data distribution (Assumption~\ref{assump:3}). The centralized server can evaluate model performance on this dataset and determine the aggregation strategy for each client. Our purpose is to make sure that a high-performance model is likely to be aggregated with models trained from a larger total amount of data. In other words, we ensure that $D_{\overline{i}}$ increases as $d_i$ increases. According to Assumption~\ref{assump:2}, a model's performance is correlated to the size of its training dataset. We model the relation between $D_{\overline{i}}$ and $d_i$ as a function $D_{others}(d_i)$. As a result, a client who contributes a better model will receive a better resultant model from the server with high probability.

\begin{algorithm}
    \caption{\textit{The Incentivized Federated Learning Algorithm.}$\ \ \ $  $S$ is the set of all clients; $T$ is the number of rounds; $E$ is the number of local epochs; $w_0$ is the common initial model; $\mathcal{D}_k$ is the local dataset of client $k$; $\mathcal{D}_v$ is the global validation dataset;  $\eta$ is the learning rate. }\label{Alg:Iter}
    \hspace*{\algorithmicindent}
    
\begin{algorithmic}[1]\\

// \textit{server executes} \\
 \textbf{procedure } SeverUpdate\\
    %\hskip1em Initialize global model to $w_0$ and send to all clients\\
    \hskip1em \textbf{for} client $k \in S$ \textbf{do}\\
    \hskip2em   $ w^k_0 \leftarrow w_0$ \\
    \hskip1em \textbf{end for}\\
    \hskip1em \textbf{for} round $t \in 1 \cdots T$  \textbf{do} \\
    \hskip2em $Acc \leftarrow \{\}$\\
    \hskip2em \textbf{for} each client $k \in S$ \textbf{in parallel do}\\
    %\hskip3em send $w^k_{t-1}$ to client $k$ \\
    \hskip3em $w^k_{t,m}$ $\leftarrow$ ClientUpdate($k$, $w^k_{t-1}$) \\
    \hskip3em evaluate $w^k_{t,m}$ on $\mathcal{D}_v$ and obtain $Acc^k$ \\
    \hskip3em $Acc \leftarrow Acc \cup \{Acc^k\}$ \\
    \hskip2em \textbf{end for} \\
    \hskip2em $Acc$ $\leftarrow$ $sorted(Acc)$\\
    \hskip2em $r(k)$ $\leftarrow$ the index function of clients according to $Acc$\\
    \hskip2em \textbf{for} client $k \in S$ \textbf{do}\\
    \hskip3em   $ w^k_t \leftarrow \sum\limits^{r(k)}_{j = 1} w^j_{t,m}$ \\
    \hskip2em \textbf{end for} \\
    \hskip1em \textbf{end for} \\
    \\
    // \textit{client $k$ executes} \\
    \textbf{procedure } ClientUpdate($k$, $w^k$) \\
    
    %\Procedure{cyclic training with KNN label correction}{}\\
    %\Input{Noisy Training Set $\tilde{S_n}$, $k$, number of epochs \textit{E}}\\
    %$\hat{S_n}\leftarrow \tilde{S_n}$\\
    \hskip1em \textbf{for} local epoch $e \in 1 \cdots \textit{E}$ \textbf{do}\\
    \hskip2em \textbf{for} batch $b \in \mathcal{D}_k $ \textbf{do}\\
    \hskip3em   $w^k= w^k -  \eta \nabla \mathcal{L}(w^k, b)$\\
    \hskip2em \textbf{end for} \\
    \hskip1em \textbf{end for} \\
    \hskip1em return $w^k$ to server
    \end{algorithmic}
\end{algorithm}

\setcounter{theorem}{0}

Hence, for vanilla FL, the derivative of the utility function with respect to the amount of training data is:
\begin{equation}
\frac{\partial u_i}{\partial d_i} = \left( \gamma \frac{\partial p_i}{\partial (d_i+D_{\overline{i}})}  - \alpha \frac{\partial  c_i}{\partial d_i} \right).
\label{eq:vanilla_derivative}
\end{equation}

It can be seen that $d_i$ has no impact on $D_{\overline{i}}$. However, for our proposed mechanism, $D_{\overline{i}}$ is depending on $d_i$:

\begin{equation}
\frac{\partial u_i}{\partial d_i} = \left( \gamma \frac{\partial p_i}{\partial (d_i+D_{\overline{i}})} + \gamma \frac{\partial p_i}{\partial (d_i + D_{\overline{i}})} \frac{\partial  D_{\overline{i}}}{\partial d_i} - \alpha \frac{\partial  c_i}{\partial d_i} \right) 
\label{eq:incentive_derivative}
\end{equation}

With the extra term, as long as $\frac{\partial D_{\overline{i}}}{\partial d_i}$ is positive, the optimal value $d^{opt}_i$ for Equation~\ref{eq:vanilla_derivative} will be greater than that of Equation~\ref{eq:incentive_derivative}.
Given $D_{\overline{i}}$ a concave function over $d_i$, a sufficiently large $\frac{\partial D_{\overline{i}}}{\partial d_i}$ guarantees that $d^{opt}_i$ equals or exceeds the maximum threshold $d^t_i$.
In this case, every client will attempt to contribute as much data as it can collect to maximize its utility.

%If this partial derivative term is large enough, which means we make $D_{others}$ very sensitive to $d$, then $d_{opt}$ exceeds the data collection threshold so that every client will attempt to contribute as much data as possible to maximize their utilities.

Lemma~\ref{lemma1} states that those clients who contribute more data can have their models learn from more data from other clients.

\begin{lemma}\label{lemma1}
If the server has access to global data distribution and can precisely measure models' performances, our mechanism ensures $\frac{\partial D_{\overline{i}}}{\partial d_i} \geq 0$ for all $i$.
\end{lemma}

\begin{comment}
    Proof tips: maybe does not need the random uniform distribution. We just need to make sure that there are sufficiently large number of clients so that increase $d$ a little bit could affect $D_{others}$

    Thoughts1: we need to model the performance between data and performance. either in detailed form or just makes assumptions about its property.

    Solution: performance monotonically increases with data, but convex

    Thoughts2: we need to model the how aggregation affects the performance. 

    Thoughts3 : we need to model how center can accurately evaluate the performance of all the client models and precisely rank the model based on performance

    For simplicity maybe just give a degeneration factor/multiplier. $p_f(d_1, d_2,...d_n) = \gamma p_T(d_1, d_2,...d_n) $ Maybe bad. Maybe we can assume the aggregation does not deteriorate the performance so that all the other clients data can be viewed as $D_{others}$

\end{comment}

\begin{proof}
    
    By Assumption~\ref{assump:2}, a model's performance is monotonically increasing as the amount of data used for training grows, so $\frac{\partial p_i}{\partial d_i} > 0$. Moreover, our aggregation strategy ensures that a higher performance model will always be aggregated with a larger set of models from other clients. Again, by Assumption~\ref{assump:2}, those models represent a larger amount of training data, which implies $\frac{\partial D_{\overline{i}}}{\partial p_i} \geq 0$. Hence, $\frac{\partial D_{\overline{i}}}{\partial d_i} = \frac{\partial D_{\overline{i}}}{\partial p_i}\cdot \frac{\partial p_i}{\partial d_i}  > 0$.
    
\end{proof}

With the Lemma~\ref{lemma1}, we can derive Theorem~\ref{thm1}. Intuitively, it states that our incentive mechanism satisfies the baseline incentive requirement that it encourages clients to contribute more data compared with the vanilla federated learning.

\iffalse
\begin{theorem} \label{thm1}
If $p(D_{total})$ is a monotonically increasing convex function and $c(d)$ is a monotonically decreasing concave function, denote the optimal data contribution as $d_{opt}$ for traditional federated learning and $d'_{opt}$ for federated learning with incentive mechanism. As long as mechanism could ensure $\frac{\partial D_{others}}{\partial d} > 0$, the client will contribute more data to maximize their utility, that is, $d'_{opt} \geq d_{opt}$.
\end{theorem}
\fi

\begin{theorem} \label{thm1}
Let $d^{opt}_i$ and $d^{opt*}_i$ denote the optimal values of $d_i$ that maximizes utility under the proposed incentive mechanism and the vanilla federated learning, respectively. The proposed mechanism always guarantees that $d^{opt}_i \geq d^{opt*}_i$.
\end{theorem}

\begin{proof}

Consider the utility functions given by Equation~\ref{eq:incentive_derivative} and Equation~\ref{eq:vanilla_derivative}.
For clarity, let $f(d_i) = \gamma \frac{\partial p_i}{\partial (d_i+D_{\overline{i}})}- \alpha \frac{\partial  c_i}{\partial d_i}$ and $g(d_i) = \gamma \frac{\partial p_i}{\partial (d_i + D_{\overline{i}})} \frac{\partial  D_{\overline{i}}}{\partial d_i}$.
Thus, the utility function for the proposed incentive mechanism is $f(d_i) + g(d_i)$ and that for vanilla federated learning is $f(d_i)$.
Note that based on Assumption~\ref{assump:4}, both $\gamma \frac{\partial p_i}{\partial (d_i+D_{\overline{i}})}$ and $- \alpha \frac{\partial  c_i}{\partial d_i}$ monotonically decreases for $d_i \in [0, \mathcal{D}_i]$, and so is $f(d_i)$. Additionally, $g(d_i) > 0$ always holds, because our incentive mechanism guarantees that $\frac{\partial  D_{\overline{i}}}{\partial d_i} \geq 0$, while Assumption~\ref{assump:2} implies that $\frac{\partial p_i}{\partial (d_i + D_{\overline{i}})} > 0$. We prove by 2 cases:
%Let $f(d) = \gamma \frac{\partial p}{\partial D_{total}}  - \alpha \frac{\partial  c}{\partial d}$.
%Since $\frac{\partial p}{\partial D_{total}}$ and $  - \alpha \frac{\partial  c}{\partial d}$ monotonically decreases with regard to $D_{total}$ and $d$ respectively. Then there are two cases:

\textbf{Case 1} ($f(d_i)\at[\big]{d_i=\mathcal{D}_i} > 0$).

This condition implies that $f(d_i) > 0$ throughout the range $d_i \in [0, \mathcal{D}_i]$. Hence, a client should devote the whole dataset to maximize its utility in both schemes, \textit{i.e.,} $d^{opt}_i = d^{opt*}_i = \mathcal{D}_i$. So $d^{opt}_i \geq d^{opt*}_i$ holds in this case.

%This case is trivial. If $f(d_t) > 0$, then $u(d)$ monotonically increases between $[ 0, d_t ]$, thus $d_{opt}=d_t$. Meanwhile, since $\frac{\partial  p}{\partial D_{total}} > 0$ based on assumptions and $\frac{\partial  D_{others}}{\partial d} > 0$ by mechanism, we have $f(d)+\frac{\partial  p}{\partial D_{total}} \frac{\partial  D_{others}}{\partial d} > 0$. Similarly, we also have $d'_{opt}= d_t$. Thus $d_{opt}=d'_{opt}$.

\textbf{Case 2} ($f(d_i)\at[\big]{d_i=\mathcal{D}_i} \leq 0$).

%In this case, we may have $f(0) < 0$, then $u(d)$ monotonically decreases, thus $d_{opt} = 0$. Then it is clearly we have $d'_{opt} \geq d_{opt}$

If $f(d_i)\at[\big]{d_i=0} \leq 0$, $f(d_i)$ is negative throughout the range. Therefore the optimal value for vanilla federated learning $d^{opt*}_i = 0$. It immediately follows that $d^{opt}_i \geq d^{opt*}_i$.

Otherwise, $f(d_i)\at[\big]{d_i=0} > 0$.  We further decompose the scenario into 2 cases. When $(f(d_i) + g(d_i))\at[\big]{d_i=\mathcal{D}_i} \geq 0$, $d^{opt}_i = \mathcal{D}_i$, and therefore $d^{opt}_i \geq d^{opt*}_i$. When $(f(d_i) + g(d_i))\at[\big]{d_i=\mathcal{D}_i} < 0$, there must exists a $d^{opt*}_i$, such that $f(d^{opt*}_i) = 0$. There must also exist a $d^{opt}_i$, such that $f(d^{opt}_i) + g(d^{opt}_i) = 0$. Because $g(d_i) > 0$ always holds, $f(d^{opt}_i) < 0$. Additionally, $f(d_i)$ is monotonically decreasing, and this indicates that $d^{opt}_i \geq d^{opt*}_i$.

\iffalse
The more interesting scenario is when $f(0) \geq 0$, then we have $f(d_{opt}) = 0$.
If we let $g(d) = \frac{\partial  p}{\partial D_{total}} \frac{\partial  D_{others}}{\partial d}$. Since  $g(d)>0$ as already discussed in Case 1, we now have $f(d_{opt}) + g(d_{opt}) = 0 + g(d_{opt})> 0$.

Again, the trivial case is when $f(d_t) + g(d_t) > 0$, where we know $d'_{opt} = d_t$ and since $d_{opt} \leq d_t$, we easily get $d'_{opt} \geq d_{opt}$.
Otherwise if $f(d_t)+g(d_t) \leq 0$, by construction, we have $f(d'_{opt}) + g(d'_{opt}) = 0$, which leads to $f(d'_{opt}) = -g(d'_{opt}) < 0$,
we then have $f(d_{opt})=0 \geq f(d'_{opt})$.
Because $f(d)$ is a monotonically decreasing function. we finally have $d'_{opt} \geq d_{opt}$.
\fi

Hence, we conclude that $d'_{opt} \geq d_{opt}$.

\end{proof}

%Theorem~\ref{thm1} provides a weak incentive guarantee for our incentive mechanism. We are more interested in the condition under which our our mechanism could ensure that every client is willing to contribute to their maximum. Therefore, we present a stronger result in Theorem~\ref{thm2}, which points out that the partial derivative  $\frac{\partial D_{others}}{\partial d}$ on the threshold point $d_t$ should be the determining factor to achieve maximum incentive. Along with our mechanism, this value depends on the distribution of data contribution among all clients.

The guarantee provided by Theorem~\ref{thm1} may not be sufficiently strong for some use cases. We would like to ensure that all clients contribute data to their maximal capabilities. We present a stronger result in Theorem~\ref{thm2}. It states that when the magnitude of the term $\frac{\partial  D_{\overline{i}}}{\partial d_i}$ is large enough to dominate the utility function, the client will devote all its data. We would like to point out that the term is also depending on the data distributions among all clients along with our incentive mechanism.

\iffalse
\begin{theorem}\label{thm2}
If $D_{others}(d)$ is monotonically increasing convex function and $\frac{\partial D_{others}}{\partial d}$ is "large" enough, e.g. as it satisfies following inequality at $d=d_t$ . then every client will take their max effort to collect data and contribute the threshold amount of data $d_t$ to participate the federated learning.
In fact, when everyone contributes the maximum data they are afford, the federated learning is at the nash equilibrium(if every contributes max data, no one will benefit by contribute less data, Note this nash equilibrium does not require convexity of $D_{others}$).
\end{theorem}
\fi

\begin{theorem}\label{thm2}
If $D_{\overline{i}}$ is a concave function over $d_i$, and $\frac{\partial  D_{\overline{i}}}{\partial d_i}$ is sufficiently large, for example, it satisfies Equation~\ref{eq:large}, all clients will contribute the maximum amount of data they possess to maximize their utilities.
\end{theorem}

\begin{equation}
\label{eq:large}
\frac{\partial  D_{\overline{i}}}{\partial d_i} \at[\big]{d_i=\mathcal{D}_i} >  \frac{ \frac{\alpha}{\gamma} \cdot \frac{\partial  c_i}{\partial d_i} \at[\big]{d_i=\mathcal{D}_i} -\frac{\partial p_i}{\partial (d_i+D_{\overline{i}})}\at[\big]{d_i=\mathcal{D}_i}} {\frac{\partial p_i}{\partial (d_i+D_{\overline{i}})}\at[\big]{d_i=\mathcal{D}_i}}.
\end{equation}

\begin{proof}
Let $u(d_i)$ and $u^*(d_i)$ denote the utility functions under the proposed incentive mechanism and the vanilla federated learning, respectively.

%Using the same reasoning as Theorem~\ref{thm1}, 
Due to Assumption~\ref{assump:4}, we have
%if we have the above equation satisfied, since $\frac{\partial D_{others}}{ \partial d} > 0$ and $\frac{\partial p}{ \partial D_{others}} > 0$,  
 
\begin{equation}
\frac{\partial^2 u_i^*}{\partial^2 d_i} = \left( \frac{\partial^2 p_i}{\partial^2 (d_i+D_{\overline{i}})}  - \alpha \frac{\partial^2 c_i}{\partial^2 d_i} \right) \leq 0.
\end{equation}

Thus $u^*(d_i)$ is a concave function within the range. Moreover, given Assumption~\ref{assump:2}, if $D_{\overline{i}}$ is a concave function over $d_i$, we have 
 
%$u(d)$ for traditional federated learning is clearly convex in that $\frac{\partial^2 p}{\partial^2 D_{total}} = \frac{\partial^2 p}{\partial^2 (d + D_{others})} \leq 0 $ and $\frac{\partial^2  c}{\partial^2 d} \geq 0 $
 
\begin{comment}
    \begin{equation}
        \frac{\partial^2 u_i}{\partial^2 d} = \left( \frac{\partial^2 p}{\partial^2 d} + \frac{\partial  p}{\partial D_{others}} \frac{\partial^2  D_{others}}{\partial^2 d} - \alpha \frac{\partial^2 c}{\partial^2 d} \right) \leq 0
    \end{equation}
        
    \begin{equation}
        \frac{\partial^2 u}{\partial^2 d} = \left( \frac{\partial^2 p}{\partial^2 d} + \frac{\partial  p}{\partial D_{others}} \frac{\partial^2  D_{others}}{\partial^2 d} +  \frac{\partial^2 p}{\partial D_{others} \partial d} \frac{\partial  D_{others}}{\partial d} - \alpha \frac{\partial^2 c}{\partial^2 d} \right) \leq 0
    \end{equation}
\end{comment}

\begin{equation}
\frac{\partial^2 u_i}{\partial^2 d_i} = \left( \gamma \frac{\partial^2 p_i}{\partial^2 (d_i+D_{\overline{i}})}(\frac{\partial  D_{\overline{i}}}{\partial d_i} + 1)^2  + \gamma \frac{\partial p_i}{\partial (d_i+D_{\overline{i}})} \frac{\partial^2  D_{\overline{i}}}{\partial^2 d_i} - \alpha \frac{\partial^2 c_i}{\partial^2 d_i} \right) \leq 0.
\label{eq:incentive_two}
\end{equation}

When Equation~\ref{eq:large}, we can immediately derive that $\frac{\partial u_i}{\partial d_i} \at[\big]{d_i=\mathcal{D}_i} > 0$. It follows that both $u(d_i)$ and $u^*(d_i)$ monotonically increase as $d_i$ increases. Thus, all clients are willing to contribute as much data as possible. 
 
%If the $D_{others}(d)$ is convex, we know $\frac{\partial^2  D_{others}}{\partial^2 d} \leq 0$, therefore, $\frac{\partial^2 u}{\partial^2 d} \leq 0$. In this case, $u(d)$ for our federated learning with incentive mechanism is also convex. Then if we have $\frac{\partial u}{\partial d_t} > 0$, we can know that $u(d)$ monotonically increases in $[0, d_t]$. Thus, the client is willing to contribute as much data as possible to maximize his utility.   
\end{proof}

%Alternatively, if $D_{\overline{i}}$ is a convex function over $d_i$, Equation~\ref{eq:incentive_two} still holds even without Equation~\ref{eq:large}. 
The proof of Theorem~\ref{thm2} requires $D_{\overline{i}}$ to be a concave function over $d_i$. Note that this is a sufficient condition, but may not be a necessary condition in practice. Additionally, we observe that this is usually the case in reality, particularly when the amounts of data possessed by clients is sub-exponential, \textit{i.e.,} a small number of clients hold a larger amount of data when the majority of the clients hold fewer data. Even without the concavity assumption, the following statement is still valid.

\newtheorem{corollary}{Corollary}[theorem]

\begin{corollary}
It is a Nash Equilibrium state that all clients contribute the maximum amount of data they possess in federated learning.
\end{corollary}

Under the proposed incentive mechanism, no clients will change their strategy when they have already devoted all their data in federated learning.

\textbf{Discussions.}

\textit{i)} The proposed approach trades off the model performance for the incentive to contribute data. Only the client who possesses the largest amount of data can obtain a model that has a similar performance as in vanilla federated learning. Therefore, if model performance is a major concern, the vanilla federated learning could be more desirable.

\textit{ii)} We made a strong assumption in Assumption~\ref{assump:2} that the training data has even quality. It can be interesting to study how data quality can affect the incentive mechanism.

\textit{iii)} The proof Theorem~\ref{thm2} requires that the magnitude of $\frac{\partial  D_{\overline{i}}}{\partial d_i}$ is sufficiently large, and the concavity of $D_{\overline{i}}$. However, both of them rely on the data distribution among the clients, albeit they are reasonable for general federated learning applications. Furthermore, we note that although these two conditions are required for the soundness of the proof, they are not always necessary for our incentive mechanism. We plan to experimentally study when these two conditions can be violated and how the proposed mechanism works in those situations. 

\section{Conclusion}

In this paper, we present an incentive mechanism for federated learning, such that all clients are willing to contribute their data. We replace monetary rewards with rewards in model performance. We proved that all clients are always more willing to contribute data under the proposed mechanism. Moreover, the proposed mechanism ensures that all clients will devote their maximum amount of data under certain conditions. Future works include assessing its effectiveness in practice and designing dedicated mechanisms for more general and diverse scenarios.

%In this paper, we presented an incentive mechanism that is proved to encourage clients to contribute more data to participate federated learning. Our main ideas is replace money reward with performance reward so that big contributors are more likely to get better models after federated aggregation. We prove the condition under which clients are willing to devote to their maximum contribution. We acknowledge that our incentive mechanism is limited due to its reliance on the client data contribution distribution to achieve the maximum contribution threshold. Future work could be to implement better diversified federated aggregation to weaken the requirement and assumption for maximum incentive achievement.

%\input{src/old_algorithm}

% keywords can be removed

\bibliographystyle{unsrt}  

%\begin{thebibliography}{1}

{\small
\bibliography{main}
}
%\end{thebibliography}
\begin{appendices}  

\section{Symbol Table}

\begin{table}[h]
\centering
\caption{Table of Notation}
\label{tab:symbol}
\begin{tabular}{ll}
\hline
\textbf{Notation}              & \textbf{Description}                       \\ \hline
FL & federated learning \\\hline
$\mathcal{C}$ & the set of all clients participate in FL \\
$\mathcal{C}_k$ & $k$-th client in $\mathcal{C}$\\
$\mathcal{D}_k$ & the dataset possessed by $\mathcal{C}_k$\\
$\mathcal{D}$ & the full dataset\\
$\mathcal{L}(w, \mathcal{D})$ & loss of model $w$ on dataset $\mathcal{D}$\\
$t$ & $t$-th round of FL, $t \in 1\cdots T$\\
$\mathcal{C}_t$ & the set of all clients who participate in or selected for round $t$\\
$N(\mathcal{C}_t)$ & the total number of samples belonging to the clients in $\mathcal{C}_t$\\
$w^k_{t,i}$ & local model of $\mathcal{C}_k$ in round $t$ and training epoch $i$\\
$w^k_{t}$ & server aggregated model for $\mathcal{C}_k$ in round $t$\\\hline
$u(\cdot),\ u_i$ & utility function of $\mathcal{C}_i$\\
$d_i$ & the portion of data possessed and contributed by $\mathcal{C}_i$\\
$d^{opt}_i$ & the value of $d_i$ which maximizes a corresponding utility function\\
$D_{\overline{i}}$ & the portion of data possessed by other clients and is available to $\mathcal{C}_i$\\
$D^{total}_{i}$ & data possessed by all clients and is available to $\mathcal{C}_i$\\
$p(\cdot),\ p_i$ & model performance function of $\mathcal{C}_i$\\
$c(\cdot),\ c_i$ & participation cost function of $\mathcal{C}_i$\\
$r(\cdot)$ & model performance index function\\
$\mathcal{D}_v$ & the global validation dataset held by the central server\\
\hline
\end{tabular}
\end{table}

\section{Partial Derivatives of $u(d_i)$}

Derivations for the partial derivatives of $u(d_i)$:

\begin{equation}
    \begin{split}
        \frac{\partial u_i}{\partial d_i} & = \gamma \frac{\partial p_i}{\partial d_i} - \alpha \frac{\partial c_i}{\partial d_i} \\
        & =  \gamma \frac{\partial p_i}{\partial (d_i+D_{\overline{i}})} \frac{\partial (d_i+D_{\overline{i}})}{\partial d_i} - \alpha \frac{\partial c_i}{\partial d_i} \\
        & =  \gamma \frac{\partial p_i}{\partial (d_i+D_{\overline{i}})} (1+\frac{\partial D_{\overline{i}} }{\partial d_i}) - \alpha \frac{\partial c_i}{\partial d_i}\\
        &= \left(\gamma  \frac{\partial p_i}{\partial (d_i+D_{\overline{i}})} + \gamma \frac{\partial p_i}{\partial (d_i + D_{\overline{i}})} \frac{\partial  D_{\overline{i}}}{\partial d_i} - \alpha \frac{\partial c_i}{\partial d_i} \right) \geq 0.
    \end{split}
\end{equation}

Denote $\frac{\partial p_i}{\partial (d_i+D_{\overline{i}})} = F(d_i+D_{\overline{i}})$, then 
\begin{equation} \label{eq2}
\begin{split}
\frac{\partial^2 u}{\partial^2 d_i} & = \gamma \left( \frac{\partial^2 p_i}{\partial (d_i+D_{\overline{i}}) \partial d_i} + \frac{\partial p_i}{\partial (d_i + D_{\overline{i}})} \frac{\partial^2 D_{\overline{i}}}{\partial^2 d_i} + \frac{\partial^2 p_i}{\partial (d_i + D_{\overline{i}}) \partial d_i} \frac{\partial D_{\overline{i}}}{\partial d_i} \right) - \beta \frac{\partial^2 c_i}{\partial^2 d_i} \\
& = \gamma \left( \frac{\partial F}{ \partial d_i} + \frac{\partial p_i}{\partial (d_i+D_{\overline{i}})} \frac{\partial^2 D_{\overline{i}}}{\partial^2 d_i} + \frac{\partial F}{ \partial d_i} \frac{\partial D_{\overline{i}}}{\partial d_i} \right) - \beta \frac{\partial^2 c_i}{\partial^2 d_i} \\
& = \gamma \left( \frac{\partial F}{ \partial (d_i+D_{\overline{i}})}(\frac{\partial D_{\overline{i}}}{\partial d_i} + 1) + \frac{\partial p_i}{\partial (d_i+D_{\overline{i}})} \frac{\partial^2 D_{\overline{i}}}{\partial^2 d_i} + \frac{\partial F}{ \partial (d_i + D_{\overline{i}})}(\frac{\partial D_{\overline{i}}}{\partial d_i} + 1) \frac{\partial D_{\overline{i}}}{\partial d_i} \right) \\
& \;\;\;- \beta \frac{\partial^2 c_i}{\partial^2 d_i} \\
&=  \gamma \frac{\partial^2 p_i}{\partial^2 (d_i+D_{\overline{i}})}(\frac{\partial  D_{\overline{i}}}{\partial d_i} +1)^2  + \gamma \frac{\partial p_i}{\partial (d_i+D_{\overline{i}})} \frac{\partial^2  D_{\overline{i}}}{\partial^2 d_i} - \alpha \frac{\partial^2 c_i}{\partial^2 d_i}  \leq 0.
\end{split}
\end{equation}

\end{appendices}

\end{document}